\documentclass[letterpaper, 10 pt, conference]{ieeeconf} 

\IEEEoverridecommandlockouts                              

\usepackage{cite}
\usepackage{amsmath,amssymb,amsfonts}
\usepackage{algorithm}
\usepackage{algpseudocode}
\usepackage{graphicx}
\usepackage{textcomp}
\usepackage[colorlinks=true,linkcolor=blue]{hyperref}
\usepackage{cuted}
\usepackage{stfloats}
 \usepackage{flushend}

\newtheorem{theorem}{Theorem}[section]

\newtheorem{proposition}[theorem]{Proposition}
\newtheorem{assumption}[theorem]{Assumption}

    \title{\LARGE \bf
Remote State Estimation with Unreliable Communication: \\ Information Asymmetry and Belief Structure
}

\author{Ioannis Tzortzis and Themistoklis Charalambous
\thanks{The authors are with the Department
of Electrical and Computer Engineering, University of Cyprus, Nicosia, Cyprus. E-mails: (\{tzortzis.ioannis,\ charalambous.themistoklis\}@ucy.ac.cy).  T. Charalambous is also with the Department of Electrical Engineering and Automation, School of Electrical Engineering, Aalto University, Espoo, 02150, Finland. }
 \thanks{This work has been partly funded by MINERVA, a European Research Council (ERC) project funded under the European Union's Horizon 2022 research and innovation programme (Grant agreement No. 101044629).}
}

\begin{document}

\maketitle

\begin{abstract}
In this paper, we study remote state estimation under unreliable forward and feedback communication between a smart sensor and a remote estimator. When the sensor cannot perfectly reconstruct the filtering state maintained by the remote estimator, an information asymmetry arises between the two sides of the network. To analyze this asymmetry, we first characterize the internal state maintained by the remote estimator and its evolution under the packet reception process. The sensor’s uncertainty about this state is then captured through a belief conditioned on its information set, and we derive a recursive update  under noisy acknowledgment feedback. Finally, we show that the resulting belief admits a finite Gaussian mixture representation whose number of components grows at most linearly over time, ensuring computational tractability. Simulation results illustrate the impact of information asymmetry on the estimation performance and demonstrate the effectiveness of the proposed framework.
\end{abstract}

\section{Introduction}
\label{sec:introduction}
Remote estimation systems operating over wireless networks are increasingly used to support monitoring and control of dynamical systems in applications such as industrial automation, smart infrastructure, and cyber–physical systems. In these settings, smart sensors equipped with local processing capabilities collect measurements  of a physical process, perform local state estimation, and transmit information to a remote estimator  through a wireless communication channel. Because wireless channels are inherently unreliable, transmitted packets may be lost due to fading, interference, or limited communication resources.   To improve reliability, many remote estimation schemes employ acknowledgment signals from the remote estimator to inform the sensor whether a transmission has been successfully received. In practice, this feedback takes the form of a one-bit acknowledgment signal (ACK/NACK). Because of its small size, this signal is often assumed to be reliably delivered, allowing the sensor to accurately track which packets have reached the remote estimator and to reconstruct its internal filtering state. However, when the feedback channel is noisy, the acknowledgment bit may be flipped during transmission. As a result, the sensor may misinterpret packet reception outcomes and can no longer accurately track the internal filtering state maintained by the remote estimator.

In this case, the sensor does not know the true packet reception outcomes and therefore cannot reconstruct the internal filtering state maintained by the remote estimator. As a result, the remote estimator and the sensor possess different information about the estimation process. In particular, the estimator maintains an internal state that is only partially known to the sensor. This mismatch between the information available at the two sides of the network creates \textit{information asymmetry}, which fundamentally alters the estimation dynamics and requires the sensor to reason probabilistically about the estimator state. 

\begin{figure}[t]
    \centering
    \includegraphics[width=\linewidth]{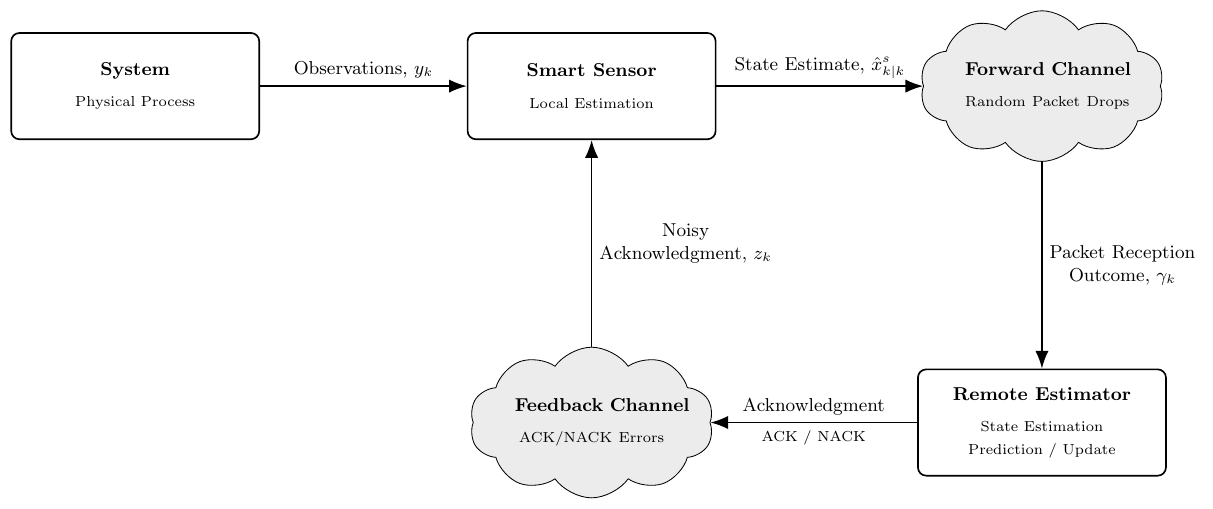}
    \caption{Remote estimation scheme with unreliable forward and feedback wireless channels.}
    \label{fig:placeholder}
\end{figure}

Remote state estimation over unreliable communication channels has been extensively studied in the networked control systems literature. Early work analyzed Kalman filtering with intermittent observations caused by packet losses and established conditions under which the estimation error remains bounded \cite{sinopoli2004kalman}. Subsequent studies investigated estimation architectures in which smart sensors perform local filtering and transmit state estimates to remote estimators over unreliable networks \cite{gupta2009data,ren2017infinite,chakravorty2019remote,Jiping-Pappas:PIMRC2025,tzortzis2025remote,makridis2026remote}. These works examined the impact of packet drops on estimation performance and the design of communication-aware estimation strategies.
Several contributions have also studied remote estimation under communication constraints such as event-triggered transmissions, scheduling policies, and limited bandwidth \cite{quevedo2011,quevedo2012state,leong2013,trimpe2014event,chen2017event,soleymani24,yang2026finite}. 
In many of these works, it is assumed that the sensor has perfect knowledge of whether its transmitted packet has been successfully received, that is, the acknowledgment feedback is assumed to be reliable. Under this assumption, the sensor can reconstruct the internal filtering state maintained by the remote estimator and remain synchronized with it.
Only a few works relax this assumption and consider imperfect acknowledgment feedback. For example, \cite{nourian2014optimal} studies optimal transmission strategies for Kalman filtering over packet-dropping links when acknowledgment signals may be erroneous. More recently, \cite{soleymani2025remote} develops a dynamic programming framework for remote estimation with unreliable forward and feedback channels and characterizes optimal coding, scheduling, and estimation strategies under communication constraints. However, existing works primarily focus on the design of transmission and estimation policies, while the structure and evolution of the sensor’s uncertainty about the estimator’s internal state under unreliable feedback are not explicitly characterized.

Motivated by the above discussion, this paper studies the estimation dynamics induced by unreliable feedback and the resulting information asymmetry between the sensor and the remote estimator. The main contributions of this paper are summarized as follows:
\begin{itemize}
    \item We characterize the internal state maintained by the remote estimator and its evolution under the packet reception process.
    \item We  model the sensor’s uncertainty about this state through a belief conditioned on the sensor’s information set and derive a recursive update  under noisy acknowledgment feedback.
    \item   We show that the resulting belief admits a finite Gaussian mixture representation whose number of components grows at most linearly over time.
    \item We derive a computationally tractable recursive update of the mixture parameters and associated weights.
\end{itemize}

The rest of the paper is organized as follows. Section \ref{sec:system_setup} presents the system setup. Section \ref{sec.problem.formulation} characterizes the estimation dynamics at the remote estimator and formulates the sensor belief over the estimator state under unreliable feedback. Section \ref{sec.belief.evolution} derives the recursive belief update and characterizes its finite-mixture representation. Section \ref{num.example} illustrates the theoretical results through a numerical example. Finally, Section \ref{sec.conclusion} concludes the paper and discusses future research directions.

\section{System Setup}
\label{sec:system_setup}

\subsection{System Model}
\label{subsec:system_model}
Consider a discrete-time  linear Gaussian system
\begin{subequations}
\begin{align}
    x_{k+1} &= A_k x_k + w_k,\\
    y_k &= C_k x_k + v_k,
\end{align}
\end{subequations}
where $x_k \in \mathbb{R}^n$ denotes the system state  and $y_k \in \mathbb{R}^m$ denotes the measured output. The matrices $A_k\in \mathbb{R}^{n\times n}$ and $C_k\in \mathbb{R}^{m\times n}$ are the state transition and observation matrices, respectively.
The process noise $w_k$ and measurement noise $v_k$ are zero-mean Gaussian random variables with covariances $Q_w$ and $Q_v$, respectively. The initial state $x_0$ is Gaussian with mean $\mu_{x_0}$ and covariance $Q_{x_0}$.
 Throughout the paper, $\mathcal{N}(\mu,\Sigma)$ denotes a Gaussian distribution with mean $\mu$ and covariance $\Sigma$.

\subsection{Smart Sensor}
\label{subsec:sensor}
The smart sensor runs a local Kalman filter based on  measurements $y_k$. 
The local state estimate and the associated error covariance evolve according to
\begin{subequations}
\begin{align}
    \hat{x}^s_{k|k-1}&=A_{k-1}\hat{x}^s_{k-1|k-1},\\
    P^s_{k|k-1}&= A_{k-1}P^s_{k-1|k-1}A_{k-1}^\top+Q_w,\\
    K_k&=P^s_{k|k-1}C_k^\top(C_kP^s_{k|k-1}C_k^\top+Q_v)^{-1},\\
    \hat{x}^s_{k|k}&=\hat{x}^s_{k|k-1}+K_k(y_k-C_k\hat{x}^s_{k|k-1}),\\
    P^s_{k|k}&=(I-K_kC_k)P^s_{k|k-1}.
\end{align}\end{subequations}
The above equations define the local filtering process at the smart sensor, whose estimate $\hat x^s_{k|k}$ is transmitted to the remote estimator over the forward wireless channel.

\subsection{Forward Wireless Channel}
\label{subsec:forward_channel}
The forward wireless channel is modeled as an unreliable communication channel characterized  by packet drops. Let $\gamma_k$ denote the packet reception outcome at time $k$, where
\begin{equation*}
    \gamma_k =
\begin{cases}
1, & \text{if the packet transmitted at time $k$ received},\\[3pt]
0, & \text{otherwise}.
\end{cases}
\end{equation*}
The sequence $\{\gamma_k\}_{k\geq 0}$ is assumed to be independent and identically distributed with
\begin{equation}\label{FWC.eq}
\operatorname{Pr(\gamma_k=1)} = 1-p,\qquad
\operatorname{Pr(\gamma_k=0)} = p, 
\end{equation}
 where $p\in [0,1]$ denotes the packet drop probability. 

At each time step, the smart sensor transmits its current estimate $\hat{x}^s_{k|k}$ to the remote estimator over the forward channel. If $\gamma_k=1$, the transmitted packet is successfully received at the remote estimator. If $\gamma_k=0$, the packet is lost and no information is delivered to the remote estimator. While in practice failed transmissions may still yield partial or unreliable information, we adopt an erasure model in which such information is not used for state reconstruction.

\subsection{Remote Estimator}
\label{subsec:remote_estimator}
The remote estimator maintains an estimate of the system state based on the information received from the forward channel. 
When a packet transmitted at time $k$ is successfully received $(\gamma_k=1)$, the remote estimator recovers a perturbed version of the transmitted estimate due to decoding and representation imperfections (e.g., quantization or residual reconstruction errors). We model this effect as an  additive decoding noise $n_k$ affecting the received estimate\footnote{This additive-noise model can be justified, for example, via dithered quantization, which yields a signal-independent quantization error that can be modeled as additive noise \cite{zamir1996information}.}. 

Since the Kalman covariance recursion depends only on the system parameters and noise statistics, $P^s_{k|k}$ can be computed locally at both the smart sensor and the remote estimator. Hence, it need not be transmitted and is assumed to be known to the remote estimator. Accordingly, the information available at the remote estimator associated with the transmission at time 
$k$ is defined as
\begin{equation*}
    {\cal S}_k =
\begin{cases}
(\hat{x}^s_{k|k}+n_{k}, P^s_{k|k}), & \text{if $\gamma_k=1$},\\[3pt]
\emptyset, & \text{if $\gamma_k=0$},
\end{cases}
\end{equation*}
where $n_k\sim \mathcal{N}(0,Q_n)$.

\begin{assumption}\label{main.ass}
The initial state $x_0$, the process noise $\{w_k\}_{k\geq 0}$, the measurement noise $\{v_k\}_{k\geq 0}$, the decoding noise $\{n_k\}_{k\geq 0}$, and the packet reception process $\{\gamma_k\}_{k\geq 0}$ are mutually independent. Moreover, each  sequence  is independent across time.
\end{assumption}

Let
\begin{equation}\label{info.RE.set}
    I_k^e=\{\gamma_{0:k-1}, {\cal S}_{0:k-1}\},
\end{equation}
denote the information available to the remote estimator at the beginning of time step $k$.
The remote estimator is initialized with $\hat{x}^e_{0|-1}=\mu_{x_0}$ and $P^e_{0|-1}=Q_{x_0}$.
Before observing the packet reception outcome $\gamma_k$, the remote estimator maintains the predicted estimate $(\hat{x}^e_{k|k-1},P^e_{k|k-1})$. After observing $\gamma_k$, it updates these quantities to obtain the filtered estimate $(\hat{x}^e_{k|k},P^e_{k|k})$. Specifically,
\begin{enumerate}
    \item Prediction step:
    Before observing $\gamma_k$, the estimator computes 
    \begin{subequations}
    \begin{align}
        \hat{x}^e_{k|k-1}&=A_{k-1}\hat{x}^e_{k-1|k-1},\\
        P_{k|k-1}^e &= A_{k-1}P^e_{k-1|k-1}A_{k-1}^\top+Q_w.
    \end{align}
        \end{subequations}
    \item Correction step (packet received, $\gamma_k=1$): If the packet is successfully received, the estimator resets its predicted estimate with the noisy estimate received from the sensor,
    \begin{equation}
        \hat{x}^{e}_{k|k} = \hat{x}^s_{k|k} + n_k, \quad
        P^e_{k|k} = P^s_{k|k} + Q_n.
    \end{equation}
    \item Correction step (packet lost, $\gamma_k=0$): If the packet is lost, the estimator retains the predicted values,
    \begin{equation}
        \hat{x}^{e}_{k|k} = \hat{x}^{e}_{k|k-1}, \qquad
        P^{e}_{k|k} = P^{e}_{k|k-1}.
    \end{equation}
\end{enumerate}

The packet reception $\gamma_k$ is then communicated to the smart sensor through the feedback channel described next.

\subsection{Feedback Wireless Channel}
\label{subsec:feeback_channel}
The packet reception outcome $\gamma_k$ is communicated to the smart sensor through a one-bit acknowledgment signal.
The feedback channel is modeled as a binary symmetric channel with crossover probability $q \in [0,1]$. Hence, the smart sensor observes $z_k \in \{0,1\}$, where
\begin{equation*}
    \operatorname{Pr}(z_k=z\mid \gamma_k=\gamma) =
\begin{cases}
1-q, & \text{if }z=\gamma,\\[3pt]
q, & \text{if }z\neq\gamma.
\end{cases}
\end{equation*}
Thus, $z_k$ provides noisy information about the packet reception outcome $\gamma_k$.

\section{Problem Formulation}\label{sec.problem.formulation}
Due to unreliable acknowledgment feedback, the  sensor does not  observe the true packet reception outcomes and therefore cannot  reconstruct the internal  state maintained by the remote estimator. We address this information asymmetry through the following tasks:
\begin{enumerate}
    \item[a)] describe the evolution of the remote estimator's internal state under packet drops and decoding noise, 
    \item[b)] derive a recursive update of the sensor’s belief over the remote estimator state based on the information available to the sensor, and 
    \item[c)] obtain a representation of this belief that remains computationally tractable over time. 
\end{enumerate}
We first characterize the evolution of the estimator state and then describe the sensor's uncertainty about it.

\subsection{Remote Estimator Information Structure}\label{subsec.RE.info}
Define the remote estimator state at the beginning of time $k$ as
\begin{equation}\label{internal.state}
    \xi^e_k = (\hat{x}^{e}_{k|k-1}, P^{e}_{k|k-1}),
\end{equation}
with state space $\Xi = \mathbb{R}^{n} \times S_+^{n}$, where $S_+^{n}$ denotes the set of symmetric positive semi-definite matrices.
Based on the estimator update rules described in Section \ref{subsec:remote_estimator}, the evolution of the  estimator state depends on the packet reception outcome $\gamma_k$. Let $F_1(\cdot)$ and $F_0(\cdot)$ denote the  update mappings corresponding to  successful packet reception and packet loss, respectively. Then,  
\begin{equation}\label{info.state.xi.recursion}
    \xi^e_{k+1} = \begin{cases}
F_0(\xi_k^e), & \text{$\gamma_k=0$},\\[3pt]
F_1(\hat{x}^s_{k|k},n_k), & \text{$\gamma_k=1$},
\end{cases}
\end{equation}
where 
\begin{equation*}
    F_0(\xi_k^e) = \Big(A_k\hat{x}^e_{k|k-1}, A_kP^e_{k|k-1}A_k^\top+Q_w\Big),
\end{equation*}
and 
\begin{equation*}
    F_1(\hat{x}^s_{k|k},n_k) = \Big(A_k(\hat{x}^s_{k|k}+n_k), A_k(P^s_{k|k}+Q_n)A_k^\top+Q_w\Big).
\end{equation*}
Under packet loss, the estimator state evolves deterministically from $\xi_k^e$, whereas under successful reception its state estimate component is stochastic due to the decoding noise.

\subsection{Smart Sensor Information Structure}
\label{subsec:SS_Info}
At the beginning of time $k$, the information available to the smart sensor is defined as
\begin{equation}\label{ss.info.set}
    I_k^s := \{y_{0:k}, z_{0:k-1}\},
\end{equation}
where $y_{0:k}$ denotes the sequence of measurements available at the sensor, and $z_{0:k-1}$ the acknowledgment signals received through the feedback  channel. The one-step delay in \eqref{ss.info.set} reflects the fact that the acknowledgment corresponding to $\gamma_k$ becomes available to the sensor at time $k+1$.

Since the smart sensor does not observe the packet reception outcomes $\gamma_k$, it cannot determine the remote estimator state $\xi_k^e$ exactly. We therefore define the sensor belief as
\begin{equation}\label{belief.dfn}
    \pi_k^s(E) := \operatorname{Pr}(\xi_k^e \in E \mid I_k^s),\quad E\subseteq \Xi,
\end{equation}
which represents the conditional distribution of the remote estimator state given the sensor's available information.

\section{Belief Evolution under Unreliable Feedback}
\label{sec.belief.evolution}
At time $k+1$, the smart sensor receives the acknowledgment $z_k$, which provides noisy information about the packet reception outcome $\gamma_k$. The  belief update is given below.
\begin{proposition}\label{proposition.belief}
For every measurable set $E\subseteq \Xi$, the sensor belief evolves according to
\begin{align}\label{eq:belief_recursion}
    \pi_{k+1}^s(E) &= \eta_k(0)\int_{\Xi}\operatorname{Pr}(\xi^e_{k+1}\in E\mid \xi_k^e=\xi,\gamma_k=0)\pi_k^s(d\xi)\nonumber\\
    &\quad +\eta_k(1)\operatorname{Pr}(\xi^e_{k+1}\in E\mid \hat{x}^s_{k|k},\gamma_k=1),
\end{align}
where
\begin{align}\label{eta.eq}
    \eta_k(\gamma)&:=\operatorname{Pr}(\gamma_k=\gamma\mid I_k^s,z_k)\nonumber\\
    &=\frac{\operatorname{Pr}(z_k\mid \gamma_k=\gamma)\operatorname{Pr}(\gamma_k=\gamma)}{\sum_{\bar{\gamma}\in \{0,1\}}\operatorname{Pr}(z_k\mid \gamma_k=\bar{\gamma})\operatorname{Pr}(\gamma_k=\bar{\gamma})}.
\end{align}
\end{proposition}
\vspace{.2cm}
\begin{proof}
    See Appendix~\ref{appendix.b}.
\end{proof}

The two terms in \eqref{eq:belief_recursion} correspond to the two possible packet reception outcomes, referred to as the loss and reception branches.  Under $\gamma_k=0$, the remote estimator state evolves through $F_0(\xi_k^e)$, and the current belief is propagated through the loss branch. In this case, 
\begin{equation*}
    \operatorname{Pr}(\xi^e_{k+1}\in E\mid \xi_k^e=\xi,\gamma_k=0) = \begin{cases}
        1, &  F_0(\xi)\in E,\\[3pt]
        0, & \text{otherwise}.
    \end{cases}
\end{equation*}
Under $\gamma_k=1$, the remote estimator state is updated through $F_1(\hat{x}^s_{k|k},n_k)$, defining the reception branch. 
Since $n_k\sim \mathcal{N}(0,Q_n)$,  $A_k(\hat{x}^s_{k|k}+n_k)\sim \mathcal{N}(A_k\hat{x}^s_{k|k}, A_kQ_nA_k^\top)$, while $P_{k+1|k}^e=A_k(P_{k|k}^s+Q_n)A_k^\top +Q_w$ is deterministic. 
Moreover, since $z_k$ is noisy, the sensor combines the two branches according to the posterior probabilities $\eta_k(0)$ and $\eta_k(1)$. Hence, \eqref{eq:belief_recursion} provides a recursive characterization of the sensor's uncertainty about the remote estimator state. The resulting belief has a hybrid structure, since the covariance component of $\xi_k^e$ takes values on a finite set determined by the packet reception process, whereas the state-estimate component evolves continuously in $\mathbb{R}^n$.

\subsection{Finite-Mixture Representation of the Sensor Belief}\label{subsec.finite.mixture}
For a fixed time $k$, each realization of the  sequence $(\gamma_0,\gamma_1,\dots,\gamma_{k-1})$ determines a branch of the remote estimator state. Since $\gamma_{\ell} \in \{0,1\}$, there are at most $2^k$ branches at time $k$. Let $J_k\in \{1,\dots,M_k\}$ denote the branch index, where $M_k$ is the number of branches maintained at time $k$. 

For branch $i$, the covariance component takes a fixed value $P^e_{k|k-1}=P_k^{(i)}$. Conditioned on $I_k^s$ and $J_k=i$, the state estimate component $\hat{x}^e_{k|k-1}\sim\mathcal{N}(\mu_k^{(i)},\Sigma_k^{(i)})$, where 
\begin{align*}
    \mu_k^{(i)}&=\mathbb{E}[\hat{x}^e_{k|k-1}\mid I_k^s, J_k=i],\\
    \Sigma_k^{(i)}&=\operatorname{Cov}(\hat{x}^e_{k|k-1}\mid I_k^s,J_k=i).
\end{align*}
Conditioned on $I_k^s$, the sensor estimates $\hat{x}^s_{\ell|\ell}$, $\ell=0,\dots,k$, are known, while $P^s_{\ell|\ell}$ is determined by the Kalman covariance recursion. Hence, for a fixed branch $J_k=i$, the remaining randomness in $\hat{x}^e_{k|k-1}$ is due to the decoding noises $n_\ell\sim \mathcal{N}(0,Q_n)$ associated with successful receptions along that branch.

Let 
\begin{equation}
    w_k^{(i)}:=\operatorname{Pr}(J_k=i\mid I_k^s),\quad i=1,\dots,M_k,
\end{equation}
denote the probability associated with branch $i$, where $w_k^{(i)}\geq 0$ and $\sum_{i=1}^{M_k}w_k^{(i)}=1$. Consequently, the integral appearing in the loss term of \eqref{eq:belief_recursion} can be written as
\begin{align}\label{finite.mixture.repr}
&\int_{\Xi}
\operatorname{Pr}(\xi^e_{k+1}\in E
\mid \xi_k^e=\xi,\gamma_k=0)\pi_k^s(d\xi)
\nonumber\\
&=
\sum_{i=1}^{M_k} w_k^{(i)}
\int_{\mathbb{R}^n}
\operatorname{Pr}\big(
\xi_{k+1}^e\in E
\mid \hat{x}^e_{k|k-1}=\hat{x},
\nonumber\\
&\qquad
P_{k|k-1}^e=P_k^{(i)},\gamma_k=0
\big)
\operatorname{Pr}(d\hat{x}\mid I_k^s,J_k=i).
\end{align}
Hence, the loss term can be decomposed over the finite set of branches. For each branch $i$, the covariance component is fixed at $P_k^{(i)}$, while the state-estimate component is Gaussian with mean $\mu_k^{(i)}$ and covariance $\Sigma_k^{(i)}$, and $w_k^{(i)}$ denotes the corresponding branch probability. Thus, the sensor belief admits a finite Gaussian-mixture representation.
The update of these parameters under the loss and reception branches is described in the next subsection.

\subsection{Belief Update Procedure}
Let the sensor belief at time $k$  be represented by
\begin{equation}
    {\cal B}_k=\Big\{w_k^{(i)}, \mu_k^{(i)}, \Sigma_k^{(i)},P_k^{(i)}\Big\}_{i=1}^{M_k},
\end{equation}
where the parameters of each branch are defined in Section \ref{subsec.finite.mixture}. We next describe how these parameters are updated to obtain ${\cal B}_{k+1}$.

Consider first the loss branch, $\gamma_k=0$. In this case, the remote estimator does not receive a new sensor estimate, and its state evolves from the previous remote state through $F_0$. Hence, each branch $i$ at time $k$ generates a corresponding loss branch at time $k+1$, with
\begin{align}
    P^{(i,0)}_{k+1}&=A_kP_k^{(i)}A_k^\top+Q_w,\label{rec.impl.a}\\
    \mu_{k+1}^{(i,0)}&=A_k\mu_k^{(i)},\\
    \Sigma_{k+1}^{(i,0)} &= A_k\Sigma_k^{(i)}A_k^\top.
\end{align}
The corresponding posterior weight is given by
\begin{equation}
    w_{k+1}^{(i,0)} = \eta_k(0)w_k^{(i)}.
\end{equation}
Consider next the reception branch, $\gamma_k=1$. In this case, the previous remote state is reset using the current sensor estimate and the decoding noise. From $F_1$,
\begin{align}
    P_{k+1}^{(1)}&= A_k(P^s_{k|k}+Q_n)A_k^\top+Q_w,\\
    \mu_{k+1}^{(1)}&=A_k\hat{x}_{k|k}^s,\\
    \Sigma_{k+1}^{(1)}&=A_kQ_nA_k^\top.
\end{align}
Unlike the loss case, these quantities do not depend on the previous branch $i$. Therefore, all branches at time $k$ lead to the same reception branch. Its posterior weight is 
\begin{equation}
    w_{k+1}^{(1)}=\eta_k(1)\sum_{i=1}^{M_k}w_k^{(i)} = n_k(1).\label{rec.impl.b}
\end{equation}

Thus, each of the $M_k$ branches generates one loss branch, while successful reception generates a single common reception branch. Consequently, $M_{k+1}\leq M_k+1.$
Moreover, if two candidate branches have identical parameters $(\mu_{k+1},\Sigma_{k+1},P_{k+1})$, they can be merged by summing their corresponding weights. Therefore, the number of joint belief components grows at most linearly with time, avoiding the exponential growth associated with all possible packet reception histories. This keeps the finite-mixture representation computationally tractable. Equations \eqref{rec.impl.a}--\eqref{rec.impl.b} provide a recursive implementation of the sensor belief update.

\section{Numerical Example}\label{num.example}
Consider a scalar linear dynamical system
\begin{align*}
    x_{k+1}&=0.9x_k+w_k,\\
    y_k&=x_k+v_k,
\end{align*}
where $x_0\sim {\cal N}(0,1)$, $w_k\sim {\cal N}(0,1)$, and $v_k\sim {\cal N}(0,1)$. The packet drop probability and feedback crossover probability are $p=0.3$ and $q=0.3$, respectively. The decoding noise is $n_k\sim {\cal N}(0,1)$. The simulation horizon is $T=50$, and  ${\cal B}_0=\{w_0=1,\mu_0=0,\Sigma_0=0,P_0=1\}$.

Fig. \ref{fig:example1} depicts the evolution of the system for a single realization of the forward and feedback channels. The first subplot shows the packet reception outcome $\gamma_k$ corresponding to the forward wireless channel. The second subplot shows the noisy acknowledgment signal $z_k$ received by the smart sensor through the feedback wireless channel. We observe that the noisy acknowledgment $z_k$ differs in many occasions from the true packet reception outcome $\gamma_k$. 
The third subplot depicts the number of belief branches $M_k$, together with the worst-case linear bound, confirming that the number of branches grows at most linearly with time.  The last subplot compares the true system state $x_k$, the predicted state maintained by the remote estimator $\hat{x}^e_{k|k-1}$, and the sensor belief mean 
\begin{equation*}
    \bar{x}^s_{k}:=\mathbb{E}_{\pi_k^s}[\hat{x}^e_{k|k-1}]=\sum_{i=1}^{M_k} w_k^{(i)}\mu_k^{(i)}.
\end{equation*}
The plot shows that the belief mean closely tracks the remote predicted state despite the unreliable feedback.

Fig. \ref{fig:example2} depicts the results of $1000$ Monte-Carlo simulations averaged over multiple realizations of the noise processes present in the system. The top subplot illustrates the remote prediction mean-squared error (MSE) between the true state  and the estimator predicted state,
\begin{equation*}
    \mbox{Remote MSE}=\mathbb{E}[(x_k-\hat{x}^e_{k|k-1})^2],
\end{equation*}
as a function of the packet drop probability $p\in [0,1]$. As expected, the prediction error increases as the packet drop probability grows, since the remote estimator receives fewer updates from the smart sensor. The bottom plot illustrates the belief tracking MSE,
\begin{equation*}
    \mbox{Belief MSE}=\mathbb{E}[(\hat{x}^e_{k|k-1}-\bar{x}^s_{k})^2],
\end{equation*}
as a function of the crossover probability $q\in [0,1]$. The proposed method is compared with a naive baseline that assumes every received acknowledgment is correct, i.e., $\gamma_k=z_k$.
For our proposed method, we distinguish the following three cases: (a) For $q=0$, we have $z_k=\gamma_k$. This corresponds to a perfect feedback channel, where the smart sensor observes the exact packet reception outcome, and therefore the belief tracking error attains its minimum value; (b) For $q\approx 0.5$, the feedback signal $z_k$ carries almost no information about $\gamma_k$. This corresponds to the case of maximum uncertainty at the sensor, resulting in the largest belief tracking error; (c) For $q=1$, we have $z_k=1-\gamma_k$. In this case, the smart sensor can  infer the packet reception outcome from the feedback signal. This represents a mirror case of $q=0$, and therefore the belief tracking error again attains its minimum value.
\begin{figure}[t]
    \centering
    \includegraphics[width=.9\linewidth]{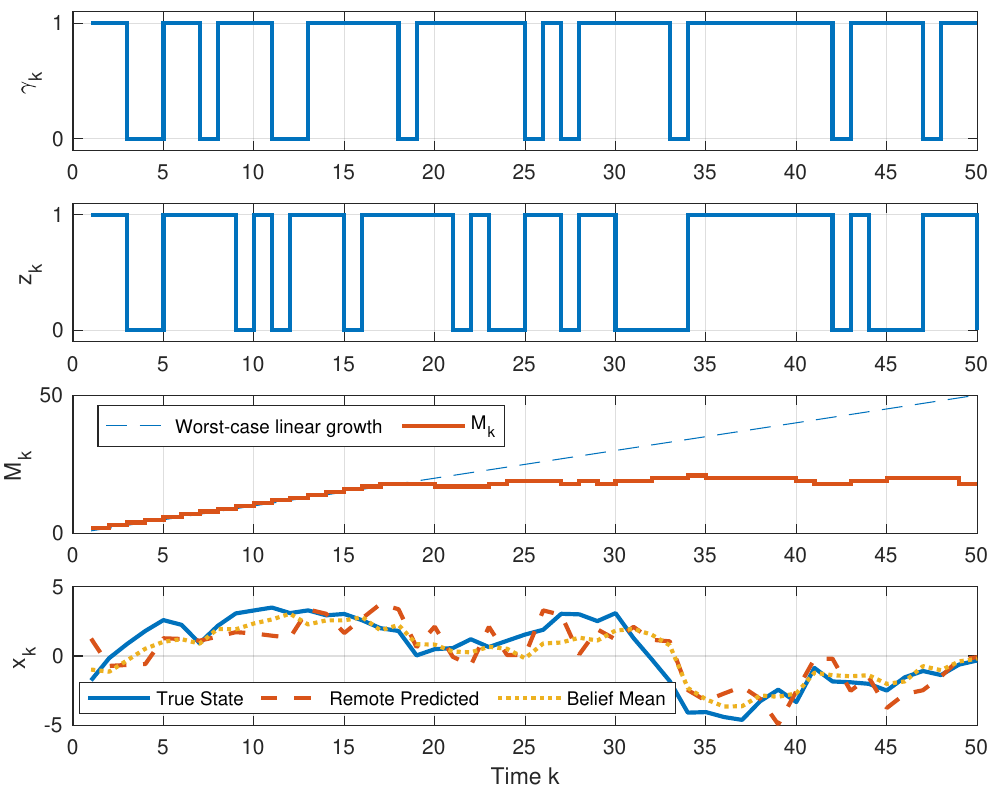}
    \caption{Evolution of the estimation process for a single realization of the forward and feedback wireless channels.}
    \label{fig:example1}
\end{figure}
\begin{figure}[t]
    \centering
    \includegraphics[width=.9\linewidth]{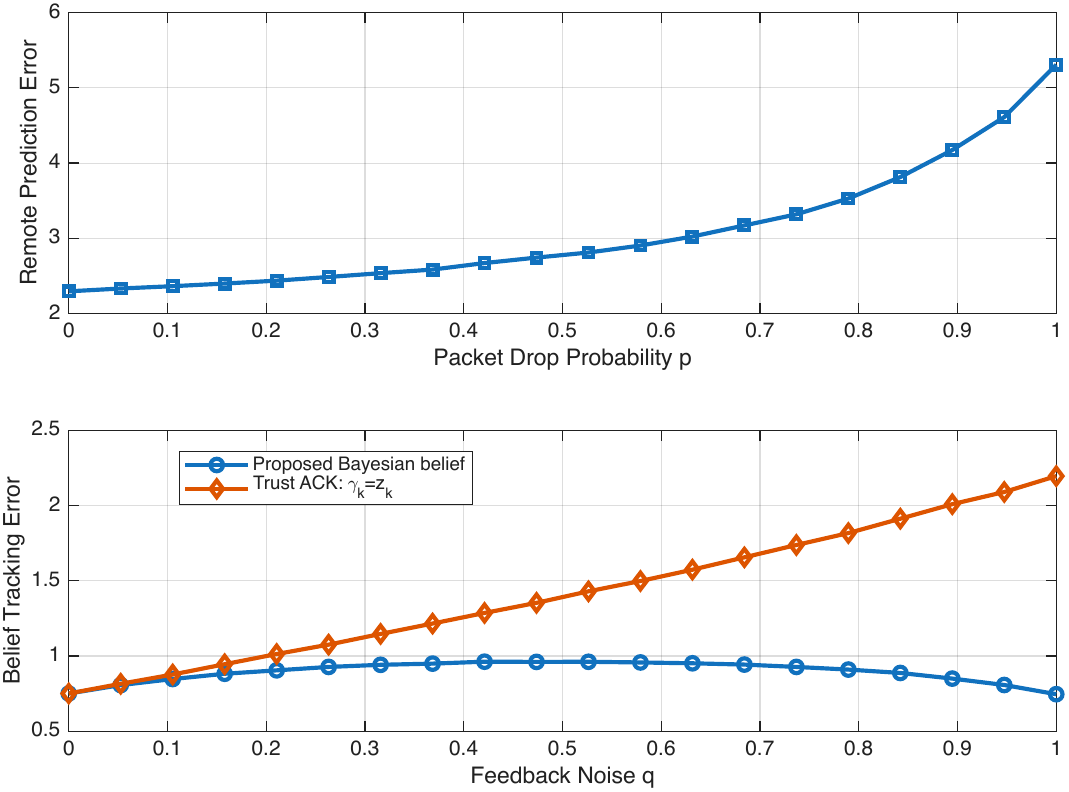}
    \caption{Monte-Carlo evaluation of the estimation performance.}
    \label{fig:example2}
\end{figure}
In contrast, the naive baseline continues to trust the received acknowledgments, and its tracking MSE increases as the crossover probability $q$ increases.   Overall, the above results illustrate the impact of unreliable communication on the estimation performance and highlight the information asymmetry between the smart sensor and the remote estimator induced by the unreliable feedback channel.

\section{Conclusions and Future Directions}\label{sec.conclusion}
This paper studies remote state estimation over wireless networks with unreliable forward and feedback communication. When acknowledgment signals are corrupted by noise, the smart sensor cannot perfectly reconstruct the state maintained by the remote estimator, leading to information asymmetry between the two sides of the network. We characterize the sensor’s uncertainty through a belief over the remote estimator state and derive its recursive evolution under noisy acknowledgment feedback. We further show that this belief admits a finite Gaussian-mixture representation whose number of components grows at most linearly with time, ensuring computational tractability.

Future work will investigate remote estimation problems in which the smart sensor  can dynamically choose when to transmit information. In such settings, the belief developed in this paper can serve as an information state for the design of optimal transmission  policies under unreliable communication.

\appendices

\section{Proof of Proposition~\ref{proposition.belief}}\label{appendix.b}
For any measurable set $E\subset \Xi$,
\begin{align*}
    \pi_{k+1}^s(E)&:=\operatorname{Pr}(\xi_{k+1}^e\in E\mid I_{k+1}^s)\\
    &=\operatorname{Pr}(\xi_{k+1}^e\in E\mid I_{k}^s, y_{k+1},z_k)\\
    &=\operatorname{Pr}(\xi_{k+1}^e\in E\mid I_{k}^s,z_k),
\end{align*}
where the last equality follows since $\xi_{k+1}^e$  and $y_{k+1}$ are conditionally independent given $I_k^s$ and $z_k$. Conditioning on the packet reception outcome $\gamma_k$ gives 
\begin{align}\label{appendix.eq1}
    &\pi_{k+1}^s(E)\\
    &= \sum_{\gamma\in \{0,1\}} \operatorname{Pr}(\xi_{k+1}^e\in E\mid I_k^s,z_k,\gamma_k=\gamma)\operatorname{Pr}(\gamma_k=\gamma\mid I_k^s,z_k)\nonumber\\
    &= \sum_{\gamma\in \{0,1\}} \operatorname{Pr}(\xi_{k+1}^e\in E\mid I_k^s,\gamma_k=\gamma)\operatorname{Pr}(\gamma_k=\gamma\mid I_k^s,z_k),\nonumber
\end{align}
where the second equality follows from the conditional independence of $z_k$ and $\xi^e_{k+1}$ given $I_k^s$ and $\gamma_k$. By Bayes' rule, the second term in \eqref{appendix.eq1} becomes 
\begin{align*}
    &\operatorname{Pr}(\gamma_k=\gamma\mid I_k^s,z_k) = \frac{\operatorname{Pr}(z_k\mid \gamma_k=\gamma)\operatorname{Pr}(\gamma_k=\gamma\mid I_k^s)}{\operatorname{Pr}(z_k\mid I_k^s)}\\
    & =\frac{\operatorname{Pr}(z_k\mid \gamma_k=\gamma)\operatorname{Pr}(\gamma_k=\gamma)}{\sum_{\bar{\gamma}\in \{0,1\}}\operatorname{Pr}(z_k\mid \gamma_k=\bar{\gamma})\operatorname{Pr}(\gamma_k=\bar{\gamma})},
\end{align*}
where we have used the fact that $\gamma_k$ is independent of $I_k^s$. Hence,
\begin{equation}\label{appendix.eq2}
    \pi^s_{k+1}(E)=\sum_{\gamma\in \{0,1\}} \eta_k(\gamma)\operatorname{Pr}(\xi_{k+1}^e\in E\mid I_k^s,\gamma_k=\gamma),
\end{equation}
where $\eta_k(\gamma)$ is defined as in \eqref{eta.eq}.

Next, we consider the two packet reception outcomes separately. For $\gamma_k=0$,
\begin{align*}
    &\operatorname{Pr}(\xi^e_{k+1}\in E\mid I_k^s,\gamma_k=0)\\
    &= \int_{\Xi} \operatorname{Pr}(\xi^e_{k+1}\in E\mid \xi_k^e=\xi,I_k^s,\gamma_k=0) \\
    &\qquad \times\operatorname{Pr}(\xi_k^e\in d\xi\mid I_k^s,\gamma_k=0)\\
    & = \int_{\Xi} \operatorname{Pr}(\xi^e_{k+1}\in E\mid \xi_k^e=\xi,\gamma_k=0)\pi_k^s(d\xi),
\end{align*}
where the second equality follows from the deterministic update $F_0(\xi_k^e)$, the conditional independence of $\gamma_k$ from $\xi_k^e$ given $I_k^s$, and the definition of $\pi_k^s$.
For $\gamma_k=1$, 
\begin{align*}
    &\operatorname{Pr}(\xi^e_{k+1}\in E\mid I_k^s,\gamma_k=1)\\
    &= \operatorname{Pr}(F_1(\hat{x}^s_{k|k},n_k)\in E\mid I_k^s,\gamma_k=1)\\
    & = \operatorname{Pr}( (A_k(\hat{x}^s_{k|k}+n_k),A_k(P^s_{k|k}+Q_n)A_k^\top+Q_w))\in E\\
    &\hspace{55mm}
{} \mid \hat{x}^s_{k|k},\gamma_k=1),
\end{align*}
where the second equality follows since $\hat{x}^s_{k|k}$ is $I_k^s$-measurable, $P^s_{k|k}$ is deterministic, and $n_k$ is independent of $I_k^s$ and $\gamma_k$. 
Substituting the above expressions into \eqref{appendix.eq2} gives the desired belief recursion \eqref{eq:belief_recursion}.

\bibliographystyle{IEEEtran}
\bibliography{autosam}

\end{document}